\documentclass[a4paper,UKenglish,cleveref, autoref,thm-restate,nolineno]{socg-lipics-v2021}

\newcommand{\Dle}{\mathrel{\le_{\mathcal D}}}   
\usepackage{tikz}
\usetikzlibrary{positioning,fit,calc,decorations.pathreplacing}

\usepackage{todonotes}
\title{From One Solution to Many: An Oracle-Based FPT Framework for Diverse Solutions under Generalized Diversity Measures} 

\titlerunning{Oracle-Based FPT Algorithms for Generalized Diversity Measures} 

\author{Pradeesha Ashok}{International Institute of Information Technology Bangalore, Bengaluru, India}{pradeesha@iiitb.ac.in}{ https://orcid.org/0000-0003-2174-0051}{}
\author{Sobyasachi Chatterjee}{The Institute of Mathematical Sciences, Chennai, India}{sobyasachichatterjee@gmail.com}{https://orcid.org/0009-0002-5878-9975}{}
\author{Soumi Nandi}{The Institute of Mathematical Sciences, Chennai, India}{nandisoumi1@gmail.com}{https://orcid.org/0000-0002-1000-991X}{}
\author{Saket Saurabh}{The Institute of Mathematical Sciences, India  \and University of Bergen, Norway}{saket@imsc.res.in}{https://orcid.org/0000-0001-7847-6402}{}
\author{Priyanshu Tiwari}{International Institute of Information Technology Bangalore, Bengaluru, India}{priyanshu.tiwari@iiitb.ac.in}{}{}
\authorrunning{P.Ashok et al.} 

\Copyright{} 



\category{} 

\relatedversion{} 

\acknowledgements{I want to thank \dots}

\nolinenumbers 

\EventEditors{John Q. Open and Joan R. Access}
\EventNoEds{2}
\EventLongTitle{42nd Conference on Very Important Topics (CVIT 2016)}
\EventShortTitle{CVIT 2016}
\EventAcronym{CVIT}
\EventYear{2016}
\EventDate{December 24--27, 2016}
\EventLocation{Little Whinging, United Kingdom}
\EventLogo{}
\SeriesVolume{42}
\ArticleNo{23}

\usepackage{amssymb,amstext,amsmath}

\usepackage{mathrsfs} 
\usepackage{xspace}
\usepackage{diagbox}
\usepackage{pdfpages}
\usepackage{graphicx}
\usepackage{algorithm,float}

\usepackage{algpseudocode}

\usepackage{nicefrac}

\usepackage{comment}
\excludecomment{dontshow}

\usepackage{mathtools}

\usepackage{tabularx}
\usepackage{multirow}
\usepackage{booktabs}
\usepackage{array}
\usepackage{tcolorbox}
\usepackage{algpseudocode}

\usepackage{rotating}

\usepackage{multicol}
\usepackage{array}

\usepackage{colortbl}
\usepackage{tabulary}
\newcolumntype{K}[1]{>{\centering\arraybackslash}p{#1}}
\usepackage{xcolor}
\usepackage{mdframed}
\usepackage{subcaption}

\newcommand{\defparprob}[4]{
\begin{tcolorbox}[colback=yellow!5!white,colframe=gray!75!black]
  \begin{tabular*}{\textwidth}{@{\extracolsep{\fill}}lr} #1   \\ \end{tabular*}
  {\bf{Input:}} #2  \\
  {\bf{Parameter:}} #3  \\  
  {\bf{Question:}} #4
  \end{tcolorbox}
}

\newcommand{\nph}{{\sf NP}-hard\xspace}

\newcommand{\fpt}{{\sf FPT}\xspace}

\newcommand{\Oh}{\ensuremath{\mathcal{O}}\xspace}

\newcommand{\I}{\ensuremath{\mathcal{I}}\xspace}

\newcommand{\yes}{\textsf{Yes}}
\newcommand{\no}{\textsf{No}}

\newcommand{\D}{\mathbb{D}}

\newcommand{\hide}[1]{}

\usepackage{cleveref}

\usepackage [autostyle, english = american]{csquotes}
\MakeOuterQuote{"}
\usepackage{graphicx,color}
\usepackage{boxedminipage}
\usepackage{framed}
\usepackage{xifthen}
\usepackage{tikz} 
\usetikzlibrary{calc,math,patterns,shapes,backgrounds}
\tikzstyle{path} = [color=black,opacity=.30,line cap=round, line join=round, line width=10pt]

\newlength{\RoundedBoxWidth}
\newsavebox{\GrayRoundedBox}
\newenvironment{GrayBox}[1]%
   {\setlength{\RoundedBoxWidth}{.93\textwidth}
    \def\boxheading{#1}
    \begin{lrbox}{\GrayRoundedBox}
       \begin{minipage}{\RoundedBoxWidth}}%
   {   \end{minipage}
    \end{lrbox}
    \begin{center}
    \begin{tikzpicture}%
       \node(Text)[draw=black!20,fill=white,rounded corners,%
             inner sep=2ex,text width=\RoundedBoxWidth]%
             {\usebox{\GrayRoundedBox}};
        \coordinate(x) at (current bounding box.north west);
        \node [draw=white,rectangle,inner sep=3pt,anchor=north west,fill=white] 
        at ($(x)+(6pt,.75em)$) {\boxheading};
    \end{tikzpicture}
    \end{center}}     

\newenvironment{defproblemx}[2][]{\noindent\ignorespaces%
                                \FrameSep=6pt%
                                \parindent=0pt%
                \vspace*{-1.5em}
                \ifthenelse{\isempty{#1}}{%
                  \begin{GrayBox}{\textsc{#2}}%
                }{%
                  \begin{GrayBox}{\textsc{#2}  parameterized by~{#1}}%
                }
                \begin{tabular*}{\textwidth}{@{\hspace{.1em}} >{\itshape} p{1.8cm} p{0.8\textwidth} @{}}%
            }{
                \end{tabular*}%
                \end{GrayBox}%
                \ignorespacesafterend
            }  

\newcommand{\DP}{\textsc{$(\mathbb{D},\Phi)$-Diverse}\xspace}
\newcommand{\msd}{\textsc{Max-Sum Diversity}\xspace}
\newcommand{\mmd}{\textsc{Max-Min Diversity}\xspace}
\newcommand{\mcd}{\textsc{Max-Coverage Diversity}\xspace}

\newcommand{\phisubfor}{\textsc{$\Phi$-Forbidden Subset}\xspace}

\newcommand{\algophiforb}{\textsc{Algo-$\Phi$-ForbSubset}\xspace}
\newcommand{\algofeasible}{\textsc{Algo-$\Phi$-Feasible}\xspace}
\newcommand{\constraindiverse}{\textsc{Constraint-Diverse $\Phi$}\xspace}
\newcommand{\mainalgo}{\textsc{DiverseSolutions}\xspace}
\newcommand{\divf}{{\D}\xspace}
\newcommand{\val}{{\sf val}\xspace}

\newcommand{\phisub}{\textsc{$\Phi$-Subset}\xspace}
\newcommand{\phisubw}{\textsc{Weighted $\Phi$-Subset}\xspace}

\begin{document}

\maketitle

\begin{abstract}
The problem of computing \emph{diverse} solutions has recently emerged as a compelling area of study, driven by
applications in fairness, robustness, and security. Rather than returning a single feasible (or optimal) solution,
the goal is to output a \emph{collection} of solutions that are meaningfully different from one another, most often
measured via symmetric differences. A growing body of work has developed algorithms for diverse variants using a
variety of techniques, including sparsification, reductions to network flows, and algebraic methods.

We study the fixed-parameter tractability of diverse variants of an implicit set-system problem.
Given parameters $k$ and $r$ and a threshold $b$, the task is to compute $r$ feasible solutions, each of size at most $k$,
whose diversity value (under a specified objective) is at least $b$.
Our main contribution is a general oracle-based meta-theorem. We identify a broad class of objectives, which we call
\emph{consistently diverse}, that subsumes standard measures such as \msd, \mmd, and \mcd. 
Assuming access to an \emph{exact empty-extension oracle}---given a forbidden set ${\sf Forb}$, it either returns a feasible
solution of a prescribed size avoiding ${\sf Forb}$ or correctly reports that none exists---we obtain an \fpt algorithm
parameterized by $k+r$. Concretely, our algorithm makes at most $(2kr)^{kr}\cdot r$ oracle calls, and in every call the oracle
parameter satisfies $s+|{\sf Forb}|\le k+2kr$.

Our framework unifies and strengthens prior oracle-based approaches to diverse solutions. In particular, compared to the
framework of Kumabe (ESA~2025), which yields a doubly-exponential bound on the number of oracle calls, our approach achieves
a single-exponential bound $2^{\Oh(kr\log(kr))}$ and directly constructs the desired tuple of solutions.
As applications, we recover \fpt algorithms for all problems covered by Kumabe’s framework with improved oracle complexity,
and we obtain strong bounds for several concrete problems, including diverse variants of classical graph and matroid problems.
\end{abstract}

\newpage
\section{Introduction}

Many computational problems admit a large number of feasible (and often optimal) solutions.
In classical algorithmic settings, the goal is to compute a single such solution.
However, in several applications, returning only one solution is insufficient: one instead seeks a
\emph{collection} of solutions that are meaningfully different from one another.
This requirement gives rise to \emph{diversity problems}.
Informally, a diversity problem asks for multiple feasible solutions whose pairwise differences,
measured using a suitable diversity objective, are large~\cite{DBLP:conf/sat/Nadel11,DBLP:conf/ijcai/PetitT15}.

The motivation for studying diversity is both practical and conceptual.
In decision-support settings, diverse alternatives offer robustness against modeling uncertainty and
allow a user to choose among qualitatively different options.
In other contexts, diversity is tied to fairness and representation, or to security and redundancy.
These motivations have driven sustained interest in computing diverse sets of solutions with provable
guarantees~\cite{barocas-hardt-narayanan,Baste2022,DBLP:conf/hotos/ForrestSA97}.

\medskip
\noindent{\bf Diversity measures.}
A common way to quantify diversity is via \emph{Hamming distance}.
For two sets $S,T\subseteq U$, their Hamming distance is the size of the symmetric difference:
\[
d_H(S,T)\ :=\ |S\triangle T|\ =\ |S\setminus T|+|T\setminus S|.
\]
Two widely studied objectives are:
(i) \msd, which maximizes the sum of pairwise distances,
$d_{\mathrm{sum}}(S_1,\ldots,S_r):=\sum_{1\le i<j\le r} d_H(S_i,S_j)$; and
(ii) \mmd, which maximizes the minimum pairwise distance,
$d_{\mathrm{min}}(S_1,\ldots,S_r):=\min_{1\le i<j\le r} d_H(S_i,S_j)$.
Recent work also considers other measures, such as \emph{Venn diversity}~\cite{drabik2025findingdiversesolutionsparameterized}
and \emph{coverage/absolute-difference} diversity~\cite{deberg2025findingdiversesolutionscombinatorial,iwamasa2025generalframeworkfindingdiverse};
the coverage objective (\mcd) is
\[
d_{\mathrm{cov}}(S_1,\ldots,S_r)\ :=\ \Bigl|\bigcup_{i\in[r]} S_i\Bigr|.
\]

\medskip
\noindent{\bf Algorithmic developments.} 
From an algorithmic perspective, substantial progress has been made. 
Diversity variants have been studied for a wide range of underlying combinatorial problems, and
efficient algorithms have been obtained using techniques such as sparsification, reductions to flow
and matching, algebraic methods, and parameterized complexity.
As a result, papers on diversity problems have appeared in several major venues, reflecting both
their theoretical richness and practical relevance~\cite{Baste2019,Baste2022,deberg2024findingdiverseminimumst,EibenKW24}.
In particular, polynomial-time and fixed-parameter algorithms (discussed below) are known for several diverse variants,
and approximation algorithms have also been investigated for \nph  cases. Polynomial-time algorithms are known for diverse variants of problems such as
\textsc{Diverse Unweighted Minimum $s$--$t$ Cut} and \textsc{Diverse Stable Matching}~\cite{iwamasa2025generalframeworkfindingdiverse,deberg2024findingdiverseminimumst} for \msd and \mcd,
and for several restricted SAT classes~\cite{MisraMR26}.
Approximation algorithms have been investigated for some \nph diverse variants, including
\textsc{Diverse Longest Common Subsequences} and \textsc{Diverse and Nice Triangulations}~\cite{ShidaPKUA24,GalvezGMPT25}.

\medskip
\noindent{\bf Our focus: FPT algorithms.}
We focus on the development of \fpt algorithms for diversity problems.
In many diverse-solution settings, the most natural parameters are the solution size $k$ (or the budgets $k_1,\ldots,k_r$),
the number of requested solutions $r$, and the diversity threshold $b$.
Accordingly, much of the literature seeks algorithms running in time
$f(k,r,b)\cdot |I|^{\Oh(1)}$ (often with $b$ implicit or bounded as a function of $k$ and $r$),
and several results are already \fpt for the combined parameter $k+r$.
Here $I$ denotes the input instance and $f$ is a computable function; in this paper, we use $k+r$ as the primary parameter. See the book by Cygan et al.~\cite{Cygan2015} for further background on parameterized algorithms.

Early work establishing \fpt tractability for diverse variants was initiated by Baste et al.~\cite{Baste2019,Baste2022}, who showed that classical tools
can often be lifted to the diverse variant: lossless kernelization yields kernels for diverse variants,
and their dynamic-programming core model implies that problems solvable by treewidth-based DP remain \fpt
when augmented with a diversity objective and parameterized by treewidth and the number of requested solutions.
This viewpoint was extended beyond treewidth by Drabik et al.~\cite{drabik2025findingdiversesolutionsparameterized},
who generalized the DP-core approach to cliquewidth for MSO$_1$-expressible problems.

Complementing these structural results, problem-driven \fpt techniques have been developed as well~\cite{DBLP:journals/mp/FominGPPS24}. 
Hanaka et al.~\cite{HanakaKKO21} used color-coding~\cite{AlonYZ95} to obtain \fpt algorithms for diverse variants of
\textsc{Diverse $k$-Path}, \textsc{Diverse Matching}, and \textsc{Diverse Subgraph Isomorphism}
(for bounded-treewidth pattern graphs) under both Max-Sum and Max-Min objectives.
Baste et al.~\cite{Baste2019} gave \fpt algorithms for \textsc{Diverse $d$-Hitting Set} and
\textsc{Diverse Feedback Vertex Set} parameterized by $k+r$, via flow-based augmentation over minimal solutions.
On the algebraic side, Eiben et al.~\cite{eiben2025determinantal} introduced \emph{determinantal sieving},
obtaining \fpt algorithms for problems such as \textsc{Diverse Bases}, \textsc{Diverse Common Independent Sets}
on linear matroids, and \textsc{Diverse Perfect Matchings} using an oracle that counts solutions modulo $2$.
These results illustrate both the breadth of the \fpt toolbox for diversification and the diversity of techniques involved,
motivating the search for unifying principles and sharper running-time guarantees. Finally, diverse variants have also been studied for {\sc Satisfiability}.
Misra et al.~\cite{MisraMR26} investigated diverse SAT variants, showing {W[1]}-hardness for \textsc{Exact Differ SAT}
on affine formulas when parameterized by $d$ and $n-d$, and on $2$-CNF formulas when parameterized by $d$,
while obtaining an \fpt algorithm for \textsc{Max Differ SAT}.
Specifically, they define the problem \textsc{Max Differ SAT} (resp. \textsc{Exact Differ SAT}) as follows: Given a Boolean formula $\Phi$ on $n$ variables, decide whether $\Phi$ has two
satisfying assignments that differ on at least (resp. exactly) $d$ variables.

\smallskip
\noindent{\bf Kumabe's unifying framework.}
Motivated by the breadth of techniques used in the \fpt literature, Kumabe~\cite{kumabe2025maxdistancesparsificationdiversificationclustering}
introduced a unifying model that captures many existing diversity problems.
Under suitable conditions on the underlying solution space, the framework yields a general-purpose theorem that
lifts algorithms for a base problem to algorithms for its \mmd  variant, providing a systematic
explanation for several previously known results.
It also yields new \fpt algorithms for problems such as \textsc{Steiner Trees}, \textsc{Minimum-Edge $s$--$t$ Flows},
and \textsc{Almost 2-SAT}.
See Section~\ref{sec:techAndcomareKumabe} for a detailed comparison with our framework.

\smallskip
\noindent{\bf Our contribution.}
In this paper, we build on and significantly extend Kumabe's framework.
We propose a new general framework for diversity problems that \emph{encompasses} a broader class of diversity measures,
\emph{generalizes} and \emph{simplifies} the structural requirements, and yields algorithms with substantially improved
running times.
In particular, compared to Kumabe's framework---which yields a doubly-exponential dependence on the parameters in the
resulting \fpt running time---our framework reduces this to a \emph{single-exponential} dependence on the parameters.
Our results subsume several known algorithms as special cases and provide a unified and conceptually simpler explanation
for why efficient algorithms for diverse solutions exist across many settings.

\subsection{Our Results, Overview, and Applications}

We need a few definitions in order to state our formalism. We borrow the notion of
\emph{implicit set systems} from~\cite{DBLP:journals/jacm/FominGLS19}. An \emph{implicit set system} is a function $\Phi$ that takes as input a binary string
$I\in\{0,1\}^*$ and outputs a set system $(U_I,\mathcal{F}_I)$, where $U_I$ is a
universe and $\mathcal{F}_I\subseteq 2^{U_I}$ is a family of subsets of $U_I$.
We refer to $I$ as an \emph{instance}, write $n:=|U_I|$ for the size of the universe,
and $|I|$ for the encoding length of the instance.

We say that $\Phi$ is \emph{polynomial-time computable} if (i) there is a polynomial-time
algorithm that, given $I$, outputs $U_I$, and (ii) there is a polynomial-time algorithm
that, given $I$, $U_I$, and a subset $S\subseteq U_I$, decides whether $S\in\mathcal{F}_I$.
Throughout the paper, all implicit set systems are assumed to be polynomial-time computable.
One may view $\mathcal{F}_I$ as the set of all feasible solutions for the instance $I$.
Typical examples include the family of feedback vertex sets of size at most $k$ in a graph,
the family of satisfying assignments of a CNF formula of weight at most $W$, and the family
of minimal hitting sets of a set system.

\smallskip 
\noindent 
{\bf From one solution to many.}
Given $(U_I,\mathcal{F}_I)$, one can ask standard questions such as: is $\mathcal{F}_I$
non-empty; can we find a set in $\mathcal{F}_I$; or what is $|\mathcal{F}_I|$?
In this paper, we go beyond finding a single solution and instead seek \emph{$r$ solutions}
from $\mathcal{F}_I$ that are as diverse as possible under a chosen diversity measure.

\smallskip
\noindent 
{\bf Diversity functions.}
A \emph{diversity function} is a mapping
$\mathbb{D}:(2^{U_I})^{r}\to\mathbb{Z}_{\ge 0}$ that assigns a non-negative value to an
$r$-tuple $\mathbb{T}=(A_1,\ldots,A_r)$ of subsets of $U_I$.
Examples include $|\bigcup_{i=1}^r A_i|$ and $\sum_{1\le i<j\le r} |A_i\triangle A_j|$. Next we formally define the problem studied in this work.

\defparprob{\DP}
{An instance $I$, integers $k_1,\ldots,k_r\in\mathbb{N}$, and a positive integer $b$; let $k:=\max_{i\in[r]} k_i$.}
{$k+r$}
{Does there exist an $r$-tuple $(X_1,\ldots,X_r)\in \mathcal{F}_I^{\,r}$ such that
$|X_i|=k_i$ for all $i\in[r]$ and $\mathbb{D}(X_1,\ldots,X_r)\ge b$?}

When it is clear from the context, we will drop the subscript $I$ and simply write
$U$ and $\mathcal{F}$. Moreover, we will denote the family of feasible solutions
$\mathcal{F}_I$ by $\mathcal{D}$, and refer to $\mathcal{D}$ as the \emph{domain of feasible solutions}.


Our first conceptual contribution is a definition of a diversity function that captures most natural notions of diversity, including \msd, \mmd, and \mcd (see \Cref{lemma:generalise diversity}). Towards this end, we first define a notion that allows us to compare certain tuples of solutions to \DP. 

\begin{definition} \label{def:less_diverse}
Let $\mathbb{T}_1=(A_1,\ldots,A_r)$ and $\mathbb{T}_2=(J_1,\ldots,J_r)$ be two $r$-tuples of sets (over a common family of feasible solutions $\mathcal D$).
We write $\mathbb{T}_1 \Dle \mathbb{T}_2$ if there exist decompositions
\[
A_i = B_i \uplus C_i
\qquad\text{and}\qquad
J_i = K_i \uplus L_i
\qquad\text{for each } i\in[r],
\]
such that the following conditions hold:
\begin{enumerate}
    \item For every $i\in[r]$, we have $K_i \subseteq B_i$ and $|A_i|=|J_i|$.
    \item For all distinct $i,j\in[r]$, $L_i \cap L_j = \emptyset$.
    \item For every $i\in[r]$, $L_i \cap B_j = \emptyset$ for all $j\in[r]$.
\end{enumerate}
\end{definition}

\begin{figure}[ht]
    \centering
    \includegraphics[width=0.5\linewidth]{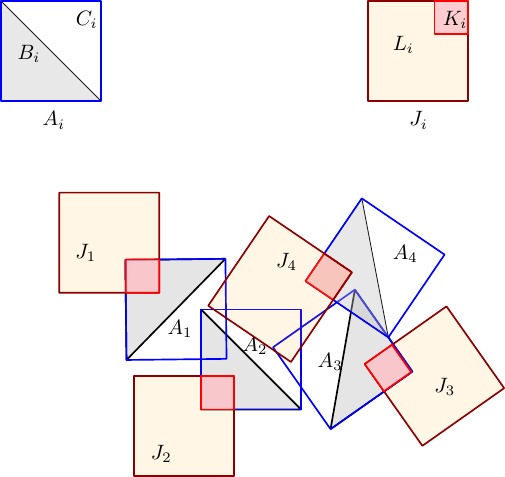}
    \caption{ Each $A_i$ consists of two disjoint parts: a grey component $B_i$ and a white component $C_i$. Similarly, each $J_i$ comprises two disjoint parts: a red component $K_i$ and a yellow component $L_i$. Let $\mathbb{T}_1=\{A_1,\dots, A_4\}$ and $ \mathbb{T}_2=\{J_1,\dots, J_4\}$. Then $\mathbb{T}_1 \Dle \mathbb{T}_2$.}
    \label{fig:def1}
\end{figure}

We interpret $\mathbb{T}_1 \Dle \mathbb{T}_2$ as saying that $\mathbb{T}_2$ is obtained from $\mathbb{T}_1$ by replacing some overlaps with \emph{fresh, non-overlapping} elements: in each coordinate $i$, we keep a subset $K_i$ from inside $A_i$, and the remaining part of $J_i$ (namely $L_i$) consists of elements that are pairwise disjoint across different coordinates and also disjoint from all core parts $B_j$. The set $C_i$ consists of those elements of $A_i$ that are the portion of $A_i$ that may be replaced by fresh, diversity-increasing elements in $J_i$.
Thus, $\mathbb{T}_2$ is ``at least as disjoint as'' $\mathbb{T}_1$ (see \cref{fig:def1} for reference).  Using this notion, we define the class of \emph{consistently diverse} diversity functions.


\begin{definition}[Consistently Diverse]\label{div_defn}
A diversity function $\mathbb{D} : (2^{U})^{r} \rightarrow \mathbb{Z}_{\ge 0}$ is \emph{consistently diverse} if for any two $r$-tuples
$\mathbb{T}_1=(A_1,\ldots,A_r)$ and $\mathbb{T}_2=(J_1,\ldots,J_r)$ of subsets of a common universe $U$, whenever $\mathbb{T}_1 \Dle \mathbb{T}_2$ we have $\mathbb{D}(\mathbb{T}_1) \le \mathbb{D}(\mathbb{T}_2)$. We call $\mathbb{D}$ \emph{polynomial-time computable} if, given any $r$ sets $S_1,\ldots,S_r \subseteq U$, the value $\mathbb{D}(S_1,\ldots,S_r)$ can be computed in time polynomial in $|U|$.
\end{definition}

For our main algorithmic result, we require access to an oracle for the following
subproblem.

\defparprob{\phisubfor}
{An instance $I$, an integer $s\in\mathbb{N}$, and a forbidden set ${\sf Forb}\subseteq U_I$.}
{$s+|{\sf Forb}|$}
{Does there exist a set $S\in \mathcal{F}_I$ such that $|S|=s$ and $S\cap {\sf Forb}=\emptyset$? If yes, output such a set $S$.}
That is, we ask for a feasible solution of prescribed size that avoids a given set of forbidden elements. Assuming that \phisubfor is fixed-parameter tractable parameterized by $s+|{\sf Forb}|$,
we obtain the following main result.

\begin{theorem}\label{thm:maxDisjoint}
Let $\mathbb{D}$ be a consistently diverse diversity function.
There is an algorithm for \DP that makes at most $(2kr)^{kr}\cdot r$
oracle calls to \phisubfor.
Moreover, in every oracle call, the parameter satisfies $s+|{\sf Forb}|\le k+2kr$.
Consequently, if \phisubfor admits an \fpt algorithm parameterized by $s+|{\sf Forb}|$,
then \DP is \fpt parameterized by $k+r$.
\end{theorem}

\noindent
This improves the parameter dependence obtained via Kumabe's framework~\cite{kumabe2025maxdistancesparsificationdiversificationclustering},
which yields a doubly-exponential bound on the number of oracle calls, namely
$2^{2^{\Oh\!\left(k\log(k\cdot r)\right)}}$. 

We note that the requirement that each solution has size exactly $k_i$ can be relaxed to
the condition $|X_i|\le k$.
This can be handled by guessing the sizes of an optimal tuple
$(X_1^\star,\ldots,X_r^\star)$ and running the algorithm for all choices
$k_1,\ldots,k_r\le k$.
The number of such choices is at most $k^{r}$.

\smallskip
\noindent{\bf Proof idea (high level).}
We solve \DP via the more general \constraindiverse formulation.
An instance of \constraindiverse is a tuple $\I=((I_1,\ldots,I_r),$ $b)$, where for each $j\in[r]$, $I_j=(I,S_j,k_j)$. 
Here $I$ is an instance of the underlying implicit set system $\Phi$ (with universe $U_I$ and feasible family $\mathcal{F}_I$),
$k_j$ is the prescribed size of the $j$-th solution, and $S_j\subseteq U_I$ is a set of elements that are
\emph{committed} to be included in the $j$-th solution (i.e., we seek $X_j\in\mathcal{F}_I$ with $|X_j|=k_j$ and $S_j\subseteq X_j$).
The task is to either output a tuple $(X_1,\ldots,X_r)$ meeting these constraints and satisfying
$\mathbb{D}(X_1,\ldots,X_r)\ge b$, or correctly report that no such tuple exists.
The original problem \DP corresponds to the special case $S_j=\emptyset$ for all $j\in[r]$.

The algorithm is recursive and is driven by a simple measure $\mu(\I)\;:=\;\sum_{j\in[r]} (k_j-|S_j|)$, which counts how many elements still remain to be chosen in total.

In each recursive call, we run a \emph{success--failure--progress} procedure.
It tries to greedily construct a tuple $(X_1,\ldots,X_r)$ by calling the oracle \phisubfor sequentially for
$i=1,\ldots,r$, each time forbidding all elements already committed elsewhere (and previously constructed solutions),
so as to enforce disjointness outside the current $S_i$.
If this greedy phase succeeds for all $r$ indices, then the resulting tuple is pairwise disjoint (outside the overlaps
already forced by the $S_j$'s), and by the definition of consistent diversity (\Cref{div_defn}) it achieves diversity at least $b$ whenever any feasible
tuple extending $(S_1,\ldots,S_r)$ achieves diversity at least $b$ (``greedy success'').

If the greedy phase fails at some index $j$, then the corresponding forbidden set ${\sf Forb}$
is a \emph{hitting set}: every feasible size-$k_j$ solution must contain an element of ${\sf Forb}$.
We therefore branch on the choice of an element $q\in{\sf Forb}$ and recurse after adding $q$ to $S_j$.
This increases $|S_j|$ by one and hence strictly decreases the measure $\mu$, guaranteeing progress.
Since initially $\mu\le \sum_{j}k_j\le kr$, the recursion depth is at most $kr$.
Moreover, the constructed hitting set satisfies $|{\sf Forb}|\le 2kr$, so the branching factor is at most $2kr$.
Each node performs at most $r$ oracle calls, yielding at most $(2kr)^{kr}\cdot r$ oracle calls overall. Finally, in each oracle call we ask for a set of size $s\le k$ that avoids a forbidden set of size at most $2kr$,
so the oracle parameter satisfies $s+|{\sf Forb}|\le k+2kr$.
Thus, if \phisubfor is \fpt parameterized by $s+|{\sf Forb}|$, then the overall algorithm is \fpt parameterized by $k+r$.

\subsection{Comparison with the Framework of Kumabe}
\label{sec:techAndcomareKumabe}

\noindent{\bf At a glance.} Kumabe sparsifies the solution domain by preserving max-distance information, whereas we
work directly with overlaps and obtain progress via guided branching on small hitting sets. 
Our approach differs conceptually from Kumabe~\cite{kumabe2025maxdistancesparsificationdiversificationclustering}
in both the \emph{object of focus} and the \emph{technical tools} used to obtain \fpt algorithms.

\smallskip
\noindent{\bf Kumabe: sparsification of the solution domain.}
Kumabe studies diversification by constructing an \emph{$r$-max-distance sparsifier}:
a carefully chosen subfamily $K\subseteq \mathcal{D}$ of the (typically exponential) domain $\mathcal{D}$ of feasible
solutions, with the guarantee that for the purpose of finding a solution that is sufficiently far from any fixed
collection of at most $r$ reference solutions, it suffices to search within $K$.
Once such a sparsifier is computed, the diversification task reduces to brute-force enumeration over $K$.

Technically, the sparsifier construction is driven by combinatorial structure in large set families,
most notably via sunflower-type arguments that justify discarding (or, equivalently, never adding) redundant solutions
while preserving the relevant distance information.
Since $\mathcal{D}$ is too large to prune explicitly, Kumabe adopts an incremental viewpoint:
start with $K=\emptyset$ and repeatedly add a witness solution from $\mathcal{D}$ until the defining property of an
$r$-max-distance sparsifier holds.
Each augmentation step is implemented via an \fpt procedure that makes oracle calls to an \emph{exact empty-extension}
type subroutine in order to (i) certify that the current $K$ already suffices, or (ii) find a violating witness solution
that must be added.
The resulting framework applies most naturally to distance-based diversity objectives such as \mmd and \msd.

\smallskip
\noindent{\bf Our framework: reasoning about overlaps and guided branching.}
In contrast, our framework does not attempt to reduce the domain of feasible solutions.
Instead, it works directly with the structure of solutions and their \emph{overlaps}.
The guiding principle is that pairwise disjoint solutions are maximally diverse for any diversity measure that rewards
disjointness; in particular, for consistently diverse functions, a tuple that is disjoint ``as much as possible''
is already optimal in the relevant monotone sense.
Algorithmically, we maintain partial commitments $(S_1,\ldots,S_r)$ and try to complete them greedily into a tuple
$(X_1,\ldots,X_r)$ using the oracle \phisubfor so as to enforce disjointness outside the already committed elements.
If this greedy phase succeeds, consistent diversity implies that we meet the threshold whenever any feasible extension
can meet it.
If it fails for some index $j$, the failure is informative: it produces a small set $Q$ that intersects \emph{every}
feasible solution for the $j$-th component, yielding a branching step.
We then recurse by guessing $q\in Q$ and adding it to $S_j$, strictly decreasing the measure
$\mu=\sum_{j\in[r]}(k_j-|S_j|)$, and thus guaranteeing progress.

\smallskip
\noindent{\bf Consequences.}
Because Kumabe’s framework proceeds by (implicitly) constructing a bounded-size sparsifier and then enumerating over it,
the resulting parameter dependence is typically doubly exponential. For example, Kumabe's framework yields a
doubly-exponential bound on the number of oracle calls, namely
$2^{2^{\Oh\!\left(k\log(k^{2}r)\right)}}$.
In contrast, our approach avoids domain sparsification and instead uses a direct progress measure and bounded branching,
yielding a single-exponential bound of
$2^{\Oh\!\left(kr\log(kr)\right)}$
on the number of oracle calls (equivalently, $(2kr)^{kr}$ up to polynomial factors).
Moreover, since our reasoning is not tied to preserving explicit pairwise distances, it applies to a broader class of
diversity measures captured by consistent diversity (including, e.g., \mmd, \msd, and coverage-style objectives),
while also yielding a conceptually simpler algorithmic template.

\subsection{Applications of \Cref{thm:maxDisjoint}}

\begin{table}[t]
\caption{Representative applications of our meta-theorem. For each problem and diversity objective, we list the previously best known \textsc{fpt} running times. \textbf{Our meta-theorem yields a uniform running time of $2^{\Oh(rk\log(rk))}$ for all problems and objectives}, suppressing polynomial factors in the input size. Here, $k$ denotes the solution size and $r$ the number of requested solutions.}
\label{tab:diversity-comparison}
\centering
\footnotesize
\setlength{\tabcolsep}{4pt}
\renewcommand{\arraystretch}{1.05}

\begin{tabular}{p{4.2cm} c c c}
\toprule
\textbf{Problem} & \textbf{Max-Sum} & \textbf{Max-Min} & \textbf{Max-Coverage} \\
\midrule
\textsc{Diverse Feedback Vertex Set} 
& $2^{\Oh(kr)}$ \cite{Baste2022} 
& $2^{\Oh(kr\log(kr))}$ \cite{Baste2022} 
& -- \\

\textsc{Diverse Odd Cycle Transversal} 
& $2^{2^{\Oh(k\log(k^2 r))}}$ \cite{DBLP:journals/orl/ReedSV04, DBLP:journals/talg/LokshtanovNRRS14, kumabe2025maxdistancesparsificationdiversificationclustering} 
& $2^{2^{\Oh(k\log(k^2 r))}}$ \cite{DBLP:journals/orl/ReedSV04, DBLP:journals/talg/LokshtanovNRRS14, kumabe2025maxdistancesparsificationdiversificationclustering} 
& -- \\

\textsc{Diverse Directed Feedback Vertex Set} 
& $2^{2^{\Oh(k\log(k^2 r))}}$ \cite{DBLP:journals/jacm/ChenLLOR08, DBLP:journals/jctb/LokshtanovRS26, kumabe2025maxdistancesparsificationdiversificationclustering} 
& $2^{2^{\Oh(k\log(k^2 r))}}$ \cite{DBLP:journals/jacm/ChenLLOR08, DBLP:journals/jctb/LokshtanovRS26, kumabe2025maxdistancesparsificationdiversificationclustering} 
& -- \\

\textsc{Diverse Disjoint Rectangle Cover} 
& $2^{2^{\Oh(k\log(k^2 r))}}$ \cite{DBLP:journals/iandc/HeggernesKLRS13, kumabe2025maxdistancesparsificationdiversificationclustering} 
& $2^{2^{\Oh(k\log(k^2 r))}}$ \cite{DBLP:journals/iandc/HeggernesKLRS13, kumabe2025maxdistancesparsificationdiversificationclustering} 
& -- \\

\textsc{Diverse Independent Set on Apex-Minor-Free Graphs} 
& $2^{2^{\Oh(k\log(k^2 r))}}$ \cite{DBLP:journals/jacm/DemaineFHT05, kumabe2025maxdistancesparsificationdiversificationclustering} 
& $2^{2^{\Oh(k\log(k^2 r))}}$ \cite{DBLP:journals/jacm/DemaineFHT05, kumabe2025maxdistancesparsificationdiversificationclustering} 
& -- \\

\textsc{Diverse Steiner Tree} 
& $2^{2^{\Oh(k\log(k^2 r))}}$ \cite{DBLP:conf/stoc/BjorklundHKK07, DBLP:journals/algorithmica/Nederlof13, kumabe2025maxdistancesparsificationdiversificationclustering} 
& $2^{2^{\Oh(k\log(k^2 r))}}$ \cite{DBLP:conf/stoc/BjorklundHKK07, DBLP:journals/algorithmica/Nederlof13, kumabe2025maxdistancesparsificationdiversificationclustering} 
& -- \\
\bottomrule
\end{tabular}
\end{table}

Since our framework employs the same type of oracle as~\cite{kumabe2025maxdistancesparsificationdiversificationclustering},
it also yields \fpt algorithms for all problems covered by their framework.
However, their approach incurs a doubly-exponential number of oracle calls in the parameters $k$ and $r$, and then
recovers explicit solutions via a brute-force search over the resulting sparsifier.
In contrast, our algorithm makes only $(2kr)^{kr}\cdot r \;=\; 2^{\Oh(kr\log(kr))}$ oracle calls and directly constructs
and outputs the desired tuple of solutions.
Consequently, for every problem that admits an \fpt oracle-based algorithm via Theorem~1.1 of
\cite{kumabe2025maxdistancesparsificationdiversificationclustering}, our framework yields a strictly smaller bound on
the number of oracle calls, while also avoiding any additional enumeration step to recover solutions.

Moreover, for several concrete problems where specialized \fpt algorithms are known---for example,
\textsc{Diverse $k$-Path}, \textsc{Diverse Matching}, and \textsc{Diverse Interval Scheduling} via color-coding~\cite{HanakaKKO21}---our
framework matches the same single-exponential dependence on the combined parameter $k+r$ in the exponent.
Likewise, for algebraic settings such as \textsc{Diverse Bases} and \textsc{Diverse Perfect Matchings},
our approach is compatible with the oracle implementations underlying determinantal sieving~\cite{eiben2025determinantal},
and yields comparable parameter dependence when instantiated with those oracles.

\section{Algorithm for \DP}\label{sec_algo}
In this section, we prove Theorem~\ref{thm:maxDisjoint} by designing an algorithm for \DP. This algorithm uses two subroutines related to \DP.

\begin{description}
  \item[Feasibility check.]
  Let \algofeasible be a procedure that verifies whether a candidate set is a feasible
  solution of the desired size. Given $(I,X,k)$, it checks whether $X\in \mathcal{F}_I$
  and $|X|=k$. It returns \yes\ if both conditions hold, and \no\ otherwise.
  We write $\algofeasible(I,X,k)$ for this call.

  \item[Oracle for \phisubfor.]
  Let $\algophiforb$ denote an algorithm for \phisubfor, which we use as an oracle.
  Given $(I,{\sf Forb},s)$, it outputs a set $S\in \mathcal{F}_I$ such that
  $|S|=s$ and $S\cap {\sf Forb}=\emptyset$. If no such set exists, it returns $\bot$.
  We write $\algophiforb(I,{\sf Forb},s)$ for this call. Throughout the algorithm, we maintain that $|{\sf Forb}| \le 2rk$.
\end{description}

Our algorithm is recursive and, in a sense, constructs the $r$ solutions in parallel so as
to maximize their ``diversity'' according to the diversity function $\divf$.
To facilitate the different steps of the algorithm, we work with a more general problem,
called \constraindiverse.
An instance of \constraindiverse consists of a tuple $(I_1,I_2,\ldots,I_r)$ and a diversity
threshold $b$, where for each $j\in[r]$, $I_j = (I, S_j, k_j)$. 
Here:
\begin{itemize}
\item $I$ is an instance of $\Phi$. 
  \item $S_j$ is a set of elements that are required (or suggested) to be included in the solution, and
  \item $k_j$ is a budget parameter.
\end{itemize}

Observe that an instance $(I,\{k_1,k_2,\ldots,k_r\},b)$ of \DP can be viewed as an instance of
\constraindiverse by setting, for each $j\in[r]$, $I_j=(I,\emptyset,k_j)$.

The output of \constraindiverse is a pair $(\val,(X_1,\ldots,X_r))$, where
$\val=\mathbb{D}(X_1,\ldots,X_r)$.
The algorithm must satisfy the following specification.
\begin{itemize}
  \item If there exists a tuple $(X_1^\star,\ldots,X_r^\star)$ such that
  $\mathbb{D}(X_1^\star,\ldots,X_r^\star)\ge b$ and $X_i^\star \supseteq S_i$ for all $i\in[r]$
  (and each $X_i^\star$ is feasible for the corresponding instance $I_i$),
  then the algorithm outputs $(\val,(X_1,\ldots,X_r))$ with $\val\ge b$.
  \item Otherwise, it outputs $(-\infty,\bot)$.
\end{itemize}

It is immediate that solving \constraindiverse on input $((I_1,\ldots,I_r),b)$ with
$I_j=(I,\emptyset,k_j)$ for all $j\in[r]$ yields a solution to \DP on input
$(I,\{k_1,\ldots,k_r\},b)$.

\begin{remark}[Permissive output]\label{rem:permissive}
Our algorithm is \emph{permissive} (in the sense common in the local-search literature~\cite{DBLP:conf/aaai/GaspersKOSS12,DBLP:journals/algorithmica/MarxS10,DBLP:conf/stoc/MoserS11}):
the returned tuple $(X_1,\ldots,X_r)$ is not required to satisfy $S_i\subseteq X_i$ for every $i$.
However, the guarantee is the following: if there exists a feasible tuple
$(X_1^\star,\ldots,X_r^\star)$ with $S_i\subseteq X_i^\star$ for all $i\in[r]$ and
$\mathbb{D}(X_1^\star,\ldots,X_r^\star)\ge b$, then the algorithm outputs some tuple
$(X_1,\ldots,X_r)$ with $\mathbb{D}(X_1,\ldots,X_r)\ge b$ (or, equivalently, outputs a value
$\val\ge b$).
\end{remark}


\noindent

\medskip

\noindent
{\bf Algorithm.}  We now present an algorithm \mainalgo that solves {\constraindiverse}. The algorithm is based on the following insight. If we can find a tuple $(X_1,\ldots,X_r)$ such that $X_i\cap X_j=\emptyset$ for all $1\le i<j\le r$,
then this is, in a natural sense, the best outcome one can hope for under any
diversity measure that rewards disjointness. In particular, since $\divf$ is
consistently diverse, any such pairwise-disjoint tuple achieves the maximum
possible diversity (among tuples of the prescribed sizes).

Recall that an instance of \constraindiverse is given by $((I_1,\ldots,I_r),b)$, where for each
$j\in[r]$, $I_j=(I,S_j,k_j)$. 
Here $I$ is an instance of $\Phi$, $S_j$ is the required/suggested set, and $k_j$ is the budget. In one recursive step of our algorithm, we run a procedure that, in \fpt time, achieves one of the following outcomes:
\begin{description}
  \item[Greedy step (success).]
  The procedure outputs a pair $(\val,(X_1,\ldots,X_r))$, where
  $\val=\mathbb{D}(X_1,\ldots,X_r)$ and $\val\ge b$, with the following permissive guarantee:
  if there exists a feasible tuple $(X_1^\star,\ldots,X_r^\star)$ extending $(S_1,\ldots,S_r)$
  such that $\mathbb{D}(X_1^\star,\ldots,X_r^\star)\ge b$, then the procedure outputs some
  tuple $(X_1,\ldots,X_r)$ with $\val\ge b$.

  \item[Terminate (failure).]
  The procedure outputs $(-\infty,\bot)$, certifying that no feasible tuple extending
  $(S_1,\ldots,S_r)$ achieves diversity at least $b$.

  \item[Recurse (progress).]
  The procedure returns a set $Q$ with $|Q|\le 2rk$ and an index $j\in[r]$ such that every
  feasible solution to $I_j=(I,S_j,F_j,k_j)$ contains at least one element of $Q$
  (i.e., $Q$ is a hitting set for the feasible-solution family of $I_j$).
\end{description}


The first two outcomes lead to termination. In the third outcome, we make progress by guessing an element
of $Q$ and adding it to $S_j$, thereby increasing $|S_j|$ by one. This guessing is implemented via recursion:
we branch on the choice of $q\in Q$ and create a new instance of \constraindiverse in which only the $j$-th
component is updated by adding $q$ to $S_j$, while all other components are kept unchanged. Formally, for each
$q\in Q$ we recurse on $\I'=((I'_1,\ldots,I'_r),b)$ where
\[
I'_j=(I,\,S_j\cup\{q\},\,k_j)
\qquad\text{and}\qquad
I'_i=I_i \ \text{ for all } i\neq j.
\]

To upper bound the running time, we associate the following measure with an instance
$\I=((I_1,\ldots,I_r),b)$ of \constraindiverse, where $I_j=(I,S_j,k_j)$ for each $j\in[r]$:
\[
~~~~~~~~~~~~~~~~~~~~~~~~~~~~\mu(\I)\;=\; \sum_{j\in[r]} \bigl(k_j - |S_j|\bigr).
\]

Observe that in each recursive step we strictly decrease the measure, i.e.,
$\mu(\mathcal{I}')<\mu(\mathcal{I})$. Hence the depth of the recursion tree is at most $kr$.
Moreover, the branching factor is at most $|Q|\le 2kr$.

\medskip

\noindent 
{\bf Success--Failure--Progress procedure.}
We now describe a procedure that, in \fpt time, returns one of three outcomes:
\emph{success} (outputs a solution tuple of diversity at least $b$),
\emph{failure} (outputs $(-\infty,\bot)$), or
\emph{progress} (returns a small set $Q$ and an index $j$ enabling a branching step).
We process the indices
$i=1,2,\ldots,r$ in order. Let $X_0:=\emptyset$. For each $i\in[r]$ we do the following.

\begin{itemize}
   \item Set
  \[
~~~~~~~~~~~~~~~~~~~~~~~~~~~~{\sf Forb}\;:=\; \left(\Bigl(\bigcup_{j\in[r]} S_j\Bigr)\ \ \cup\ \ \Bigl(\bigcup_{j<i} X_j\Bigr)\right) \setminus  S_i.
\]
  Invoke the oracle $\algophiforb(I,{\sf Forb},k_i)$.
  If it returns a set $X_i$, we continue to the next index.
  Otherwise, we output $Q:={\sf Forb}$ and $j:=i$.
\end{itemize}

First, consider the case in which the \texttt{for}-loop outputs $Q:={\sf Forb}$ and $j:=i$. In this case, the correctness argument is as follows. 
 The oracle call for index $i$ fails precisely when there is
no feasible set of the required size for $I_i$ that avoids ${\sf Forb}$. Equivalently, \emph{every}
feasible solution $X$ of size $k_i$ for $I_i$ (regardless of whether it contains $S_i$) intersects
${\sf Forb}$. Thus ${\sf Forb}$ is a hitting set for the family of feasible size-$k_i$ solutions of
$I_i$, and returning $Q:={\sf Forb}$ with $j:=i$ is correct; branching on an element of $Q$ is therefore safe.

The bound $|Q|\le 2kr$ is immediate. Recall that
  \[~~~~~~~~~~~~~~~~~~~~~~~~~~~~
{\sf Forb}\;:=\; \left(\Bigl(\bigcup_{j\in[r]} S_j\Bigr)\ \ \cup\ \ \Bigl(\bigcup_{j<i} X_j\Bigr)\right) \setminus  S_i.
\]
Since $|S_j|\le k$ for all $j$ and $|X_j|\le k$ for all $j<i$, we have
\[
|{\sf Forb}|
\ \le\
\sum_{j\in[r]\setminus\{i\}} |S_j|
\;+\;
\sum_{j<i} |X_j|
\ \le\
(r-1)k + (i-1)k
\ \le\ 2(r-1) k \ \le\
2kr.
\]

Now assume that the \texttt{for}-loop produces a tuple $(X_1,\ldots,X_r)$. We claim that if
$\val=\divf(X_1,\ldots,X_r)\ge b$, then we can return $(\val,(X_1,\ldots,X_r))$. Otherwise,
this branch is a \no-instance (i.e., no feasible tuple can achieve diversity at least $b$).  The correctness of this step is formalized in Lemma~\ref{lem:greedy-success}. The three outcomes of this procedure will be referred to as \emph{greedy success}, \emph{failure}, and \emph{progress}, respectively.

\medskip
\noindent\textbf{Base case.} 
To complete the recursive algorithm, we now describe the base case. If $\mu(\mathcal{I})\le 0$, then the solutions are completely determined by the sets
$S_j$. In this case, the algorithm invokes the polynomial-time feasibility checker
$\algofeasible$ to verify each candidate set $S_j$. If some $S_j$ is infeasible, the algorithm
returns $(-\infty,\bot)$, as the current branch cannot yield a valid solution tuple.
Otherwise, it computes the diversity value $\mathbb{D}(S_1,\ldots,S_r)$  and returns the tuple 
$(\val=\mathbb{D}(S_1,\ldots,S_r),(S_1,\ldots,S_r))$ if this value is at least $b$; if not, it again
returns $(-\infty,\bot)$. Next we give the correctness proof and the running-time analysis.

\subsection{Correctness and Time Complexity Analysis}

We now establish the correctness of the recursive procedure
\mainalgo. We begin by introducing the basic terminology of \emph{extend} and \emph{avoid}, which we  use throughout the analysis.

\begin{definition}[Extend and Avoid]\label{def:extend-avoid}
Let $(X'_1,\ldots,X'_r)$ and $(S_1,\ldots,S_r)$ be two sequences of $r$ sets. We say that $(X'_1,\ldots,X'_r)$ \emph{extends} $(S_1,\ldots,S_r)$ if
$X'_i \supseteq S_i$ for every $i\in[r]$. For two sets $A$ and $B$, we say that $A$ \emph{avoids} $B$ if
$A\cap B=\emptyset$.
\end{definition}

\noindent
The algorithm maintains the following invariant throughout its execution.
\smallskip

\noindent{\bf \em Correctness invariant.}
Consider any recursive call on an instance $\mathcal{I}=((I_1,\ldots,I_r),b)$ with current
partial choices $(S_1,\ldots,S_r)$.
If there exists a feasible tuple $(X_1^\star,\ldots,X_r^\star)$ that extends $(S_1,\ldots,S_r)$ and satisfies
$\mathbb{D}(X_1^\star,\ldots,X_r^\star)\ge b$, then the algorithm either
\begin{enumerate}[(\roman*)]
  \item outputs a tuple $(X_1,\ldots,X_r)$ with $\mathbb{D}(X_1,\ldots,X_r)\ge b$, or
  \item returns a branching pair $(Q,j)$ such that there exists some $q\in Q$ for which the
  recursive call obtained by replacing $S_j$ with $S_j\cup\{q\}$ remains consistent with such a witness tuple.
\end{enumerate}


This invariant is formalized by the lemmas stated below, which show that the algorithm
reconstructs a sufficiently diverse tuple of solutions whenever one exists. The following lemma establishes the correctness of the \emph{greedy success} step.




\begin{lemma}[Correctness of greedy success, $\star$]\footnote{The proofs labeled with $\star$ can be found in the appendix.}\label{lem:greedy-success}
Let $((I_1,\ldots,I_r),b)$ be an instance of \constraindiverse, where
$I_j=(I,S_j,k_j)$ for each $j\in[r]$.
Suppose that \mainalgo successfully constructs sets $(X_1,\ldots,X_r)$ in the greedy step,
that is, the greedy construction succeeds for all $r$ indices.
If there exists a tuple $(X_1^\star,\ldots,X_r^\star)$ such that
\begin{itemize}
  \item \textbf{Extension.} $(X_1^\star,\ldots,X_r^\star)$ extends $(S_1,\ldots,S_r)$,
  \item \textbf{Feasibility and cardinality.} $X_j^\star\in \mathcal{F}_I$ and $|X_j^\star|=k_j$
  for all $j\in[r]$, and
  \item \textbf{Diversity.} $\D(X_1^\star,\ldots,X_r^\star)\ge b$,
\end{itemize}
then the tuple $(X_1,\ldots,X_r)$ returned by the greedy phase satisfies
$\D(X_1,\ldots,X_r)\ge b$. 
\end{lemma}

We now state the main lemma that will be used to prove the correctness of \Cref{thm:maxDisjoint}.

\begin{lemma}[$\star$]\label{lemma:inductionMaxDisjoint}
Let $((I_1,\ldots,I_r),b)$ be an instance of \constraindiverse, where $I_j=(I,S_j,k_j)$ for each
$j\in[r]$. Suppose there exists a tuple of feasible solutions $(X_1^\star,\ldots,X_r^\star)$ such that
\begin{itemize}
  \item \textbf{Extension.} $(X_1^\star,\ldots,X_r^\star)$ extends $(S_1,\ldots,S_r)$,
  \item \textbf{Feasibility and cardinality.} $X_j^\star\in \mathcal{F}_I$ and $|X_j^\star|=k_j$ for all $j\in[r]$, and
  \item \textbf{Diversity.} $\D(X_1^\star,\ldots,X_r^\star)\ge b$.
\end{itemize}
Then \mainalgo returns a tuple $(X_1,\ldots,X_r)$ such that
\begin{itemize}
  \item $\D(X_1,\ldots,X_r)\ge b$, and
  \item $X_i\in\mathcal{F}_I$ and $|X_i|=k_i$ for all $i\in[r]$.
\end{itemize}
\end{lemma}

Finally, we combine Lemmas~\ref{lem:greedy-success} and~\ref{lemma:inductionMaxDisjoint} to prove Theorem~\ref{thm:maxDisjoint}.

\begin{proof}[Proof of Theorem~\ref{thm:maxDisjoint}]
We invoke \mainalgo on the instance $((I_1,\ldots,I_r),b)$, where
$I_j=(I,\emptyset,k_j)$ for all $j\in[r]$.
If the algorithm returns a tuple with diversity at least $b$, we output \yes;
otherwise, we output \no.
Correctness follows from \Cref{lemma:inductionMaxDisjoint}.

\medskip
\noindent
{\bf Running-time analysis.}
Initially,
\[
\mu(\mathcal{I})=\sum_{j\in[r]} (k_j-|\emptyset|)=\sum_{j\in[r]} k_j \le rk.
\]
Each progress step increases $|S_j|$ by one for some $j$, and therefore decreases the measure by one.
Hence the recursion depth is at most $rk$. In any progress step, the procedure returns a set
$Q$ with $|Q|\le 2rk$, and we branch on each element of $Q$, so the branching factor is at most $2rk$.
Moreover, each recursive call performs at most $r$ oracle invocations to $\algophiforb$ (one per index in the greedy phase).
Therefore, the total number of oracle calls is at most $r\cdot (2rk)^{rk}$, and the overall running time is fixed-parameter tractable.
\end{proof}

\begin{remark}[Improved branching]\label{rem:improved-branching}
With a more careful implementation/analysis of the greedy phase, one can ensure that the returned
hitting set $Q$ satisfies $|Q|\le rk$ (instead of $2rk$). Consequently, the branching factor becomes
at most $rk$, and the bound on the total number of oracle calls to $\algophiforb$ improves to $r\cdot (rk)^{rk}$.
\end{remark}

\section{Applications}
\label{prelims}
In this section, we show how to derive concrete running-time bounds for diverse variants of several combinatorial problems.
To this end, we first demonstrate that our notion of a consistently diverse function captures the standard diversity objectives studied in the literature.
We then explain how to implement the oracle subproblem \phisubfor for a broad class of natural combinatorial problems, yielding explicit \fpt running times via our meta-theorem.

\subsection{Diversity Functions}
We begin by observing that several standard diversity objectives—such as \mmd, \msd, and \mcd—are captured as special
cases of our generalized notion (Definition~\ref{div_defn}). In particular, each of these objectives is
\emph{consistently diverse}.

\begin{lemma}[$\star$]
\label{lemma:generalise diversity}
The diversity measures \msd, \mmd, and \mcd are consistently diverse and polynomially computable.
\end{lemma}

\subsection{Algorithms for \phisubfor}
For ease of presentation, we restate the oracle problem.

\defparprob{\phisubfor}
{An instance $I$, an integer $s\in\mathbb{N}$, and a forbidden set ${\sf Forb}\subseteq U_I$.}
{$s+|{\sf Forb}|$}
{Does there exist a set $S\in \mathcal{F}_I$ such that $|S|=s$ and $S\cap {\sf Forb}=\emptyset$? If yes, output such a set $S$.}

When ${\sf Forb}=\emptyset$, we denote the resulting problem by \phisub.
Note that \phisub is simply the underlying (``non-diverse'') feasibility problem of finding a feasible set of size $s$.

We will also use a weighted variant, denoted \phisubw, in which each element of $U_I$ has a nonnegative weight and the goal is to find
a feasible set of weight $s$ (or of weight at most $s$, depending on the base problem).
Then \phisubfor reduces to \phisubw as follows: assign weight $s+1$ to every element in ${\sf Forb}$ and weight $1$ to every element in $U_I\setminus{\sf Forb}$.
Any minimum-weight feasible solution of size $s$ will avoid ${\sf Forb}$ whenever such a solution exists, and otherwise no feasible solution of size $s$ that avoids ${\sf Forb}$ exists. Thus, for several problems---including \textsc{Odd Cycle Transversal}, \textsc{Directed Feedback Vertex Set}, and \textsc{Multicut}---this yields the required algorithms for \phisubfor (via their corresponding weighted formulations), and hence \fpt algorithms for the diverse variants through our meta-theorem~\cite{DBLP:journals/talg/KimKPW24,
DBLP:journals/csr/KratschPSW26, DBLP:journals/jacm/KimKPW25,
DBLP:journals/siamcomp/KimKPW25}.

However, for many cases---especially cut problems and parity-type problems such as
\textsc{Odd Cycle Transversal}, \textsc{Directed Feedback Vertex Set}, and \textsc{Multicut}---we can often avoid
weights altogether. Instead, one can reduce \phisubfor to the corresponding \emph{unweighted} base problem by applying
standard graph-modification operations (e.g.\ torso-style reductions, gadget substitutions, or equivalent transformations)
that enforce the forbidden set by construction. This yields a direct implementation of the oracle in the unweighted
setting, while preserving the relevant parameter dependence~\cite{DBLP:journals/talg/LokshtanovPSSZ20,DBLP:journals/siamcomp/MarxR14}.

For {\sc Set-Cover} type problems (including \textsc{Hitting Set}) which is known to be \fpt parameterised by solution size , implementing the forbidden set is immediate:
we simply delete all forbidden elements from the ground set (and update the family accordingly) and then run the
underlying algorithm on the resulting instance.

Thus, in most implementations, we do not need to treat $|{\sf Forb}|$ as an explicit parameter: the forbidden set can be enforced
by a simple preprocessing step (deleting forbidden elements, or modifying the instance so that forbidden elements are unusable)
and then invoking an algorithm for the underlying problem parameterized only by $s$ (or $k$).
This is convenient in practice and can potentially be exploited to obtain stronger running-time bounds.

\section{Conclusion}\label{conc}
We conclude by emphasizing the generality and strength of our results.
Our framework generalizes and strengthens existing approaches to computing diverse solutions by identifying the
\emph{exact empty-extension oracle} as a unifying algorithmic primitive.
In particular, we show that whenever a problem admits such an oracle, a broad class of its diversity variants is
fixed-parameter tractable with respect to the combined parameter $k+r$.
Our results apply to a generalized diversity measure that subsumes \msd, \mmd, and \mcd, thereby capturing several
prominent objectives studied in the literature under a single umbrella.
At the same time, our techniques do not imply fixed-parameter tractability for \emph{Venn diversity};
establishing whether Venn diversity admits an \fpt algorithm in this generality remains an interesting open problem.

As a concrete open problem, it would be interesting to design a $2^{\Oh(kr)}$-time algorithm for
\textsc{Diverse Odd Cycle Transversal} parameterized by $k+r$.
Our current framework yields a running time of $2^{\Oh(kr\log(kr))}$ (up to polynomial factors in the input size).
Another natural direction is to investigate kernelization: in particular, establishing whether
\textsc{Diverse Odd Cycle Transversal} admits a polynomial kernel (parameterized by $k+r$) would be very interesting.



\bibliography{ref}
\section{Appendix}

\subsection{Proof of Lemma~\ref{lem:greedy-success}}

\begin{proof}
Let $(X^\star_1,\ldots,X^\star_r)$ be a tuple satisfying the assumptions.
Since it extends $(S_1,\ldots,S_r)$, for each $i\in[r]$ we can write
$X_i^\star = S_i \uplus R_i^\star$ for some $R_i^\star$.

For the tuple $(X_1,\ldots,X_r)$ produced by the greedy phase, define $\hat S_i := X_i\cap S_i$ and
$R_i := X_i\setminus \hat S_i$. Then $X_i=\hat S_i \uplus R_i$ for all $i\in[r]$.

We now show that $(X_1^\star,\ldots,X_r^\star)\Dle (X_1,\ldots,X_r)$ by
Definition~\ref{def:less_diverse}. Consider the decompositions
\[
A_i:=X_i^\star,\quad B_i:=S_i,\quad C_i:=R_i^\star,
\qquad\text{and}\qquad
J_i:=X_i,\quad K_i:=\hat S_i,\quad L_i:=R_i.
\]
We verify the required conditions:
\begin{enumerate}[(\roman*)]
  \item For each $i\in[r]$, we have $K_i=\hat S_i\subseteq S_i=B_i$. Moreover, $|X_i|=k_i$ by construction
  and $|X_i^\star|=k_i$ by assumption, hence $|A_i|=|J_i|$.

  \item For distinct $i,j\in[r]$, we have $L_i\cap L_j = R_i\cap R_j=\emptyset$.
  This follows from the next claim.
  \begin{claim}
      For every $i\in[r]$ and every $t<i$, we have
  \[
  R_i \cap X_t = \emptyset .
  \]
  \end{claim}
  \begin{claimproof}
      When constructing $X_i$, the oracle is invoked with a forbidden set containing
  $\bigl(\bigcup_{t<i} X_t\bigr)\setminus S_i$. Since $R_i\cap S_i=\emptyset$ by definition of $R_i$,
  any element chosen into $R_i$ lies outside $S_i$ and therefore cannot lie in any earlier $X_t$.
  \end{claimproof}

  Now fix $i\neq j$ and assume w.l.o.g.\ that $j<i$. Since $R_j\subseteq X_j$, the claim implies
  $R_i\cap R_j\subseteq R_i\cap X_j=\emptyset$.

  \item For all $i,j\in[r]$, we have $L_i\cap B_j = R_i\cap S_j=\emptyset$. For $j=i$ this holds by the
  definition of $R_i$. For $j\neq i$, the forbidden set used in the oracle call for index $i$ contains
  $\bigcup_{j\neq i} S_j$, hence no element of $S_j$ can be selected into $R_i$.
\end{enumerate}
Therefore, $(X_1^\star,\ldots,X_r^\star)\Dle (X_1,\ldots,X_r)$.

Finally, since $\D$ is consistently diverse (Definition~\ref{div_defn}), we obtain
\[
\D(X_1,\ldots,X_r)\ \ge\ \D(X_1^\star,\ldots,X_r^\star)\ \ge\ b,
\]
which completes the proof.
\end{proof}

\subsection{Proof of Lemma~\ref{lemma:inductionMaxDisjoint}}

\begin{proof}
We prove the lemma by induction on the measure
\[
\mu(\mathcal{I}) \;:=\; \sum_{j\in[r]} \bigl(k_j-|S_j|\bigr),
\]
where $\mathcal{I}=((I_1,\ldots,I_r),b)$ and each $I_j=(I,S_j,k_j)$.

\smallskip
\noindent\textbf{Base case.}
Assume $\mu(\mathcal{I})\le 0$. Since the lemma assumes the existence of a witness tuple
$(X_1^\star,\ldots,X_r^\star)$ that extends $(S_1,\ldots,S_r)$ with $|X_j^\star|=k_j$ for all $j$,
we have $|S_j|\le k_j$ for every $j\in[r]$. Hence each summand $k_j-|S_j|\ge 0$, and
$\mu(\mathcal{I})\le 0$ implies $k_j-|S_j|=0$ for all $j$, i.e., $|S_j|=k_j$ for all $j$.
Because $X_j^\star$ extends $S_j$ and has size $k_j$, it follows that $X_j^\star=S_j$ for all $j$.
Therefore, returning $(S_1,\ldots,S_r)$ yields diversity at least $b$, and the claim holds.

\smallskip
\noindent\textbf{Induction hypothesis.}
Assume the statement holds for every instance $\mathcal{I}'=((I'_1,\ldots,I'_r),b)$ with
$I'_j=(I,S'_j,k_j)$ such that $\mu(\mathcal{I}')<\mu(\mathcal{I})$, provided there exists a witness
tuple $(X_1^\star,\ldots,X_r^\star)$ that extends $(S'_1,\ldots,S'_r)$, satisfies $|X_j^\star|=k_j$
for all $j\in[r]$, and has $\D(X_1^\star,\ldots,X_r^\star)\ge b$.

\smallskip
\noindent\textbf{Induction step.}
Assume $\mu(\mathcal{I})>0$. Let $j'\in\{0,1,\ldots,r\}$ be the largest index such that the greedy
phase successfully constructs $X_1,\ldots,X_{j'}$. If $j'=r$, then the greedy phase succeeds for all
indices and the claim follows from \Cref{lem:greedy-success}. Thus assume $j'<r$ and let $i:=j'+1$.

Let
\[
X \;:=\; \bigcup_{t=1}^{j'} X_t
\]
be the set of elements selected so far. The greedy phase fails at index $i$, meaning that the oracle
call for $I_i$ fails with forbidden set
\[
{\sf Forb}
\;:=\;
\Bigl(\Bigl(\bigcup_{j\in[r]} S_j\Bigr)\ \cup\ X\Bigr)\setminus S_i .
\]
Equivalently, there is \emph{no} feasible solution $Y\in\mathcal{F}_I$ of size $k_i$ such that
$Y\cap{\sf Forb}=\emptyset$.

On the other hand, by assumption there exists a witness tuple $(X_1^\star,\ldots,X_r^\star)$ extending
$(S_1,\ldots,S_r)$ with $\D(X_1^\star,\ldots,X_r^\star)\ge b$. We claim that
\[
X_i^\star \cap {\sf Forb}\neq \emptyset .
\]
Indeed, if $X_i^\star\cap{\sf Forb}=\emptyset$, then $X_i^\star$ is a feasible size-$k_i$ solution
that avoids ${\sf Forb}$, contradicting the failure of the oracle call. Since ${\sf Forb}$ is defined
with $S_i$ removed, any element of $X_i^\star\cap{\sf Forb}$ automatically lies in $X_i^\star\setminus S_i$.
Thus there exists an element $x\in{\sf Forb}$ such that $x\in X_i^\star\setminus S_i$.

The algorithm branches on each choice $q\in{\sf Forb}$ by forming a new instance
$\mathcal{I}_q=((I_{1,q},\ldots,I_{r,q}),b)$ in which only the $i$-th component is updated:
\[
I_{i,q}=(I,\,S_i\cup\{q\},\,k_i),
\qquad\text{and}\qquad
I_{\ell,q}=I_\ell\ \text{ for all }\ell\neq i.
\]
Consider the branch corresponding to $q=x$. Then $(X_1^\star,\ldots,X_r^\star)$ extends the updated
tuple of suggested sets (since $x\in X_i^\star$), and the measure strictly decreases:
\[
\mu(\mathcal{I}_x)
\;=\;
\mu(\mathcal{I})-1
\;<\;
\mu(\mathcal{I}).
\]
Therefore, $\mathcal{I}_x$ satisfies the induction hypothesis, and so the recursive call on this
branch returns a tuple $(X'_1,\ldots,X'_r)$ with $\D(X'_1,\ldots,X'_r)\ge b$ and $|X'_j|=k_j$ for all
$j\in[r]$. Since the algorithm explores all branches $q\in{\sf Forb}$, it will eventually explore the
branch $q=x$ and thus return a tuple of diversity at least $b$.

This completes the induction.
\end{proof}

\subsection{Proof of Lemma~\ref{lemma:generalise diversity}}
\begin{proof}
Recall that for two sets $A$ and $B$, the notation $A \triangle B$ denotes their symmetric difference.
Let $A_i,B_i,C_i,J_i,K_i,L_i$ for each $i \in [r]$ be as in Definition~\ref{def:less_diverse}.
We begin by establishing a basic inequality on pairwise symmetric differences.

\begin{claim}
For any $i,j \in [r]$, we have $|A_i \triangle A_j| \le |J_i \triangle J_j|$.
\end{claim}

\begin{claimproof}
Recall that for any two sets $S, T$, the size of their symmetric difference is given by $|S \triangle T| = |S| + |T| - 2|S \cap T|$. 
Since $|A_i| = |J_i|$ and $|A_j| = |J_j|$, it suffices to show that $|J_i \cap J_j| \le |A_i \cap A_j|$. 

Expanding the intersection $J_i \cap J_j$ using the decomposition $J_i = K_i \uplus L_i$:
\[
J_i \cap J_j = (K_i \cap K_j) \cup (K_i \cap L_j) \cup (L_i \cap K_j) \cup (L_i \cap L_j).
\]
By the conditions of Definition~\ref{def:less_diverse}:
\begin{itemize}
    \item $L_i \cap L_j = \emptyset$ for $i \neq j$.
    \item $L_i \cap B_j = \emptyset$, and since $K_j \subseteq B_j$, we have $L_i \cap K_j = \emptyset$.
    \item Similarly, $K_i \cap L_j = \emptyset$.
\end{itemize}
Thus, $J_i \cap J_j = K_i \cap K_j$. Since $K_i \subseteq B_i \subseteq A_i$ and $K_j \subseteq B_j \subseteq A_j$, we have:
\[
|J_i \cap J_j| = |K_i \cap K_j| \le |B_i \cap B_j| \le |A_i \cap A_j|.
\]
This confirms $|A_i \triangle A_j| \le |J_i \triangle J_j|$.
\end{claimproof}
Using the above claim, we now show that \msd and \mmd are consistently diverse.

\begin{itemize}
\item For \msd, we have
\[
\begin{aligned}
d_{\mathrm{sum}}(A_1,\ldots,A_r)
&= \sum_{1 \le i < j \le r} |A_i \triangle A_j| \\
&\le \sum_{1 \le i < j \le r} |J_i \triangle J_j| \\
&= d_{\mathrm{sum}}(J_1,\ldots,J_r).
\end{aligned}
\]

\item For \mmd, we obtain
\[
\begin{aligned}
d_{\mathrm{min}}(A_1,\ldots,A_r)
&= \min_{1 \le i < j \le r} |A_i \triangle A_j| \\
&\le \min_{1 \le i < j \le r} |J_i \triangle J_j| \\
&= d_{\mathrm{min}}(J_1,\ldots,J_r).
\end{aligned}
\]
\end{itemize}

It remains to consider \mcd.
By Condition~2 of Definition~\ref{def:less_diverse}, the sets $L_i$ for $i \in [r]$ are pairwise disjoint.
Moreover, by Conditions~1 and~3, each $L_i$ is disjoint from $\bigcup_{j \in [r]} K_j$.
Consequently, we obtain
\[
\begin{aligned}
d_{\mathrm{cov}}(J_1,\ldots,J_r)
&= \left| \bigcup_{i \in [r]} K_i \right|
   + \sum_{i \in [r]} (|J_i| - |K_i|) \\
&\ge \left| \bigcup_{i \in [r]} K_i \right|
   + \left| \bigcup_{i \in [r]} (A_i \setminus K_i) \right| \\
&\ge d_{\mathrm{cov}}(A_1,\ldots,A_r).
\end{aligned}
\]

Finally, all three measures are polynomially computable.
Both \msd and \mmd can be computed by evaluating symmetric differences over all pairs of sets,
while \mcd can be computed by taking the union of all sets and computing its size.
\end{proof}

\end{document}